\documentclass[%
 reprint,
 showpacs,
 amsmath,amssymb,amsfonts,
 aps,pra,
]{revtex4-1}

\usepackage[utf8]{inputenc}
\usepackage[T1]{fontenc}
\usepackage[english]{babel}

\usepackage{graphicx}
    \graphicspath{{./figs/}}

\usepackage{amsthm,dsfont}
\usepackage{braket}

\newcommand{\hilb}[1]{\mathcal{H}_{#1}}
\newcommand{\ctr}[1]{{{#1}^\dagger}}
\newcommand{\PQ}[0]{\mathcal{P}^Q_d}
\newcommand{\PQd}[1]{\mathcal{P}^Q_{#1}}
\newcommand{\mPQ}[0]{\mathcal{P}^Q_{\max}}
\newcommand{\PD}[0]{\mathcal{P}^\delta}

\newtheorem{definition}{Definition}
\newtheorem{theorem}{Theorem}
    \theoremstyle{remark}
\newtheorem*{remark}{Remark}

\begin{document}


\title{Quantumness and the role of locality on quantum correlations}

\author{G. Bellomo}
\affiliation{Instituto de Física La Plata (IFLP), CONICET, and Departamento de Física, Facultad de Ciencias Exactas, Universidad Nacional de La Plata, 1900 La Plata, Argentina}
\author{A. Plastino}%
\affiliation{Instituto de Física La Plata (IFLP), CONICET, and Departamento de Física, Facultad de Ciencias Exactas, Universidad Nacional de La Plata, 1900 La Plata, Argentina}
\author{A. R. Plastino}
\affiliation{CeBio and Secretaría de Investigaciones,
Universidad Nacional del Noroeste de la
Provincia de Buenos Aires (UNNOBA-CONICET),
R. Saenz Peña 456, Junín, Argentina}
\date{\today}

\begin{abstract}
Quantum correlations in a physical system are usually studied with respect to a unique (fixed) decomposition of the system into subsystems, without fully exploiting the rich structure of the state-space. Here, we show several examples in which the consideration of different ways to decompose a physical system enhances the quantum resources and accounts for a more flexible definition of quantumness-measures. Furthermore, we give a new perspective regarding how to reassess the fact that local operations play a key role in general quantumness-measures that go beyond entanglement ---as discord-like ones. We propose a way to quantify the maximum quantumness of a given state. Applying our definition to low-dimensional bipartite states, we show that different behaviours are reported for separable and entangled states than those corresponding to the usual measures of quantum correlations. We show that there is a close link between our proposal and the criterion to witness quantum correlations based on the rank of the correlation matrix, proposed by Daki\'c, Vedral and Brukner.
\end{abstract}



\maketitle

\section{Introduction}
Contemporary physics' recent technological and theoretical progress shows that quantum computation is a feasible and not-so-far-away perspective (see Refs.~\cite{Feng13,*Ren14,*Brow14,*Zhan14,*Zu14,*Mong14,*Howa14,*Goss15} and references therein). Novelties would bring significant improvement in the performance of information processing tasks, and the main ingredient involved resides in (quantum) correlations that cannot be implemented with classical systems. Hence, the study of quantum correlations have been one of the most pursued issues in quantum physics of the last decade (see, e.g., the excellent reviews of the Horodeckis~\cite{Horo09} and Modi \textit{et al.}~\cite{Modi12}). Quantum Entanglement and Quantum Discord (QD) are two of the main families of quantum correlations measures, which are closely related to the way in which a system can be decomposed as a mixture of product states. A non-entangled (or separable) state $\rho^{AB}_{sep}$ over the Hilbert space $\hilb{A}\otimes\hilb{B}$, with respect to the bipartition $A|B$, can be written as a convex combination of product states as $\rho^{AB}_{sep}=\sum_{k}{p_k\rho^A_k\otimes\rho^B_k}$, with $p_k\geq0$ and ${\sum_k{p_k}=1}$. In turn, a classically-correlated (CC) state $\rho^{AB}_{clas}$, with respect to the same bipartition, can be expressed as a mixture of local orthogonal projectors as $\rho^{AB}_{clas}=\sum_{ij}{p_{ij}\ket{i^A}\bra{i^A}\otimes\ket{j^B}\bra{j^B}}$, where $p_{ij}\geq0$, ${\sum_{ij}{p_{ij}}=1}$ and $\braket{i^{A(B)}|i'^{A(B)}}=\delta_{ii'}$, with $0\leq i\leq dim(\hilb{A(B)})$. That is, $\rho^{AB}_{clas}$ is diagonal in a product basis $\{\ket{i^A}\otimes\ket{j^B}\}$. A state that is not CC, is said to be quantum-correlated (QC).

What is the main difference between a non-correlated (or product) state, $\rho^{AB}_{prod}=\rho^A\otimes\rho^B$, and a CC one (as $\rho^{AB}_{clas}$ given above), with regards to their quantum capabilities? One may suspect that a CC state is as useless as a product one when performing an information task that necessarily involves quantum resources. However, {\it this is not true}. Let us consider that one has $\rho^{AB}_{clas}$ and also one has access to other local degrees of freedom, i.e. that our initial \textit{system $+$ environment} state is ${\rho^{ext}=\eta^{\bar{A}}\otimes\rho^{AB}_{clas}\otimes\eta^{\bar{B}}}$, where $\eta^{\bar{A}}$ ($\eta^{\bar{B}}$) depicts the state of the environmental degrees of freedom in $A$ ($B$). \textit{Now, it is easy to show that there are local observables with respect to whom the state is quantumly correlated.} It suffices to notice that $\rho^{A'B'}=\text{Tr}_{env}{U\rho^{ext}\ctr{U}}$ is, in general, QC with respect to $A'|B'$, where `$env$' denotes the environment degrees of freedom and $U$ denotes a local unitary operation (LU) that respects that local bipartition, i.e. $U=U^{\bar{A}A}\otimes U^{B\bar{B}}$, and accounts for the inspection of local observables.

Consideration of different observables of a quantum system leads to alternative descriptions, and quantum correlations are relative to such observables-election. Zanardi~\cite{Zana01} first noticed the effect of this relative character vis-a-vis the quantum entanglement of multiqubit states and proposed a formalization under a general algebraic framework~\cite{Zana04}. Later, Barnum \textit{et al.}~\cite{Barn04} gave a subsystem-independent notion of entanglement. Harshmann and Ranade~\cite{Hars11} proved that all pure states of a finite-dimensional (and unstructured) Hilbert space are equivalent as entanglement resources in the ideal case that one has complete access and control of observables. Given that CC implies separability and given that the question about separability becomes relative to the preferred observables (the ones that determine the local subsystems), the question about the correlations on CC states becomes relative too. It is worth noting that these ideas have been successfully applied, for example, to the investigation of quantum phase transitions~\cite{Wu04,Somm04,Cakm15} and to quantum entanglement in systems of indistinguishable particles~\cite{Bena10,*Bala13a,*Bala13b,*Iemi14,*Kill14}.

In this work, we focus on the less studied situation of mixed states under a locality restriction: we allow only \textit{local} unitary operations (over the enlarged Hilbert space) in order to explore the observables' subspaces of each local subsystem. In the pure state scenario, global unitary operations lead to equivalence 
regarding quantum correlations (in that case, entanglement). As expected, mixedness and locality impose some restrictions on the achievable quantum correlations when considering the mentioned relative character (see Appendix~\ref{sec:QCunderU} for a discussion on the role of mixedness on discord-like measures under global unitaries, for states in $\mathbb{C}^4$).

We adopt here the distinction between \textit{quantum correlations} and \textit{quantumness} of correlations, previously discussed by Giorgi \textit{et al.} in terms of genuine/non-genuine quantum correlations~\cite{Gior11} and by Gessner \textit{et al.}~\cite{Gess12}. It is interesting to note that Ollivier and Zurek, in their seminal paper~\cite{Olli01}, have already coined the idea that QD accounts for the quantumness of correlations, and not to the amount of quantum correlations \textit{per se}.

The paper is organized as follows. In section~\ref{sec:QC&sub} we discuss the main result, namely the way in which general quantum correlations depend on the subsystem-decomposition of a given quantum system. In section~\ref{sec:PQ} we advance a proposal for quantifying the mentioned effect from an information-theoretic perspective, which we will call Potential Quantumness, and we prove several interesting properties of such measure. Finally, in section~\ref{sec:PD} we specialize our study to discord-like correlations and display some of their features when applied to simple low-dimensional models. Section~\ref{sec:Conc} is devoted to a summary and conclusions.

\section{Quantumness and subsystems} \label{sec:QC&sub}
Let us give a concrete example. We begin with a CC state of two qubits as $\rho^{AB}_{clas}=p\ket{0^A}\bra{0^A}\otimes\ket{0^B}\bra{0^B}+(1-p)\ket{1^A}\bra{1^A}\otimes\ket{1^B}\bra{1^B}$, with $0\leq p\leq1$, where $\{\ket{i}\}_{i=0,1}$ is the standard (computational) basis. If we have access to one auxiliary qubit on each location, we can set the extended state to be $\rho^{ext}=\ket{0^{\bar{A}}}\bra{0^{\bar{A}}}\otimes\rho^{AB}_{clas}\otimes\ket{0^{\bar{B}}}\bra{0^{\bar{B}}}$. Now, any LU operation $U^{\bar{A}A}\otimes U^{B\bar{B}}$ accounts for different partitions of the local subsystems $\bar{A}|A$ ($B|\bar{B}$) into new subsystems $\bar{A'}|A'$ ($B'|\bar{B'}$). For example, if $U^{\bar{A}A}=U_{\text{cH}} U_{\text{S}}=(U^{B\bar{B}})^\dagger$, where $U_{\text{cH}}$ is a controlled Hadamard gate and $U_{\text{S}}$ a swap gate, the transformed reduced state is $\rho^{A'B'} = p\ket{0^{A'}}\bra{0^{A'}}\otimes\ket{0^{B'}}\bra{0^{B'}}+(1-p)\ket{+^{A'}}\bra{+^{A'}}\otimes\ket{+^{B'}}\bra{+^{B'}}$, where $\ket{\pm^{A'(B')}}=\frac{1}{\sqrt{2}}(\ket{0^{A'(B')}}\pm\ket{1^{A'(B')}})$. The new state, $\rho^{A'B'}$, is not CC anymore. Thus, we have revealed some hidden quantumness of $\rho^{AB}_{clas}$, just by considering a transformation over local degrees of freedom. This can not be done if the state is uncorrelated: for $\rho^{AB}_{prod}$, it is straightforward to show that the same procedure gives a new uncorrelated state $\rho^{A'B'}_{prod}$. \textit{This result clearly distinguishes $\rho^{AB}_{clas}$ from $\rho^{AB}_{prod}$ with regards to quantum information processing capabilities.} Such feature holds for every non-product (i.e. correlated) state: \textit{if $\rho^{AB}$ is a bipartite correlated state and $A$ and/or $B$ has local access to auxiliary degrees of freedom, then it is possible to find quantum correlations between $\rho^{AB}$.} We take this result to be our first theorem.

    \begin{theorem} \label{theo:qc1}
Let $\rho^{AB}$ be a non-product density operator over $\hilb{A}\otimes\hilb{B}$. Let $\eta^{\bar{A}}$ and $\eta^{\bar{B}}$ be the `ready' states of two ancillary systems. Then, for the extended state $\eta^{\bar{A}}\otimes\rho^{AB}\otimes\eta^{\bar{B}}$ it is possible to find a subsystem decomposition that preserves the local bipartition and possess quantum correlations.
    \end{theorem}
    
    \begin{proof}
The statement can be proved straightforwardly. If $\rho^{AB}$ is non-product state then it can be CC or QC. If it is QC then there is no need to extend our system, it already possess quantum correlations. If it is CC, we can choose the auxiliary states to be pure, $\eta^{\bar{A}}=\eta^{\bar{B}}=\ket{0}\bra{0}$. Then, over the extended state $\ket{0}\bra{0}\otimes\rho^{AB}\otimes\ket{0}\bra{0}$, we can apply LU operations $U_{\bar{A}A}\otimes U_{\bar{B}B}$ that corresponds to different decompositions of each part ($A$ and $B$) into subsystems. Finally, tracing out the auxiliary degrees of freedom results in a modified state $\rho^{A'B'}$. The action of the unitaries over the reduced state is equivalent to a local quantum trace-preserving operation (see, e.g., Ref.~\cite{Niel00}). But, performance of arbitrary local channels can convert any CC state into a QC one~\cite{Daki10}. This observation ends the proof. Note, however, that for product states there is not a local operation that correlates both parts, neither quantum not even classically.
    \end{proof}

    \begin{figure}
    \centering
    \includegraphics[width=.77\columnwidth]{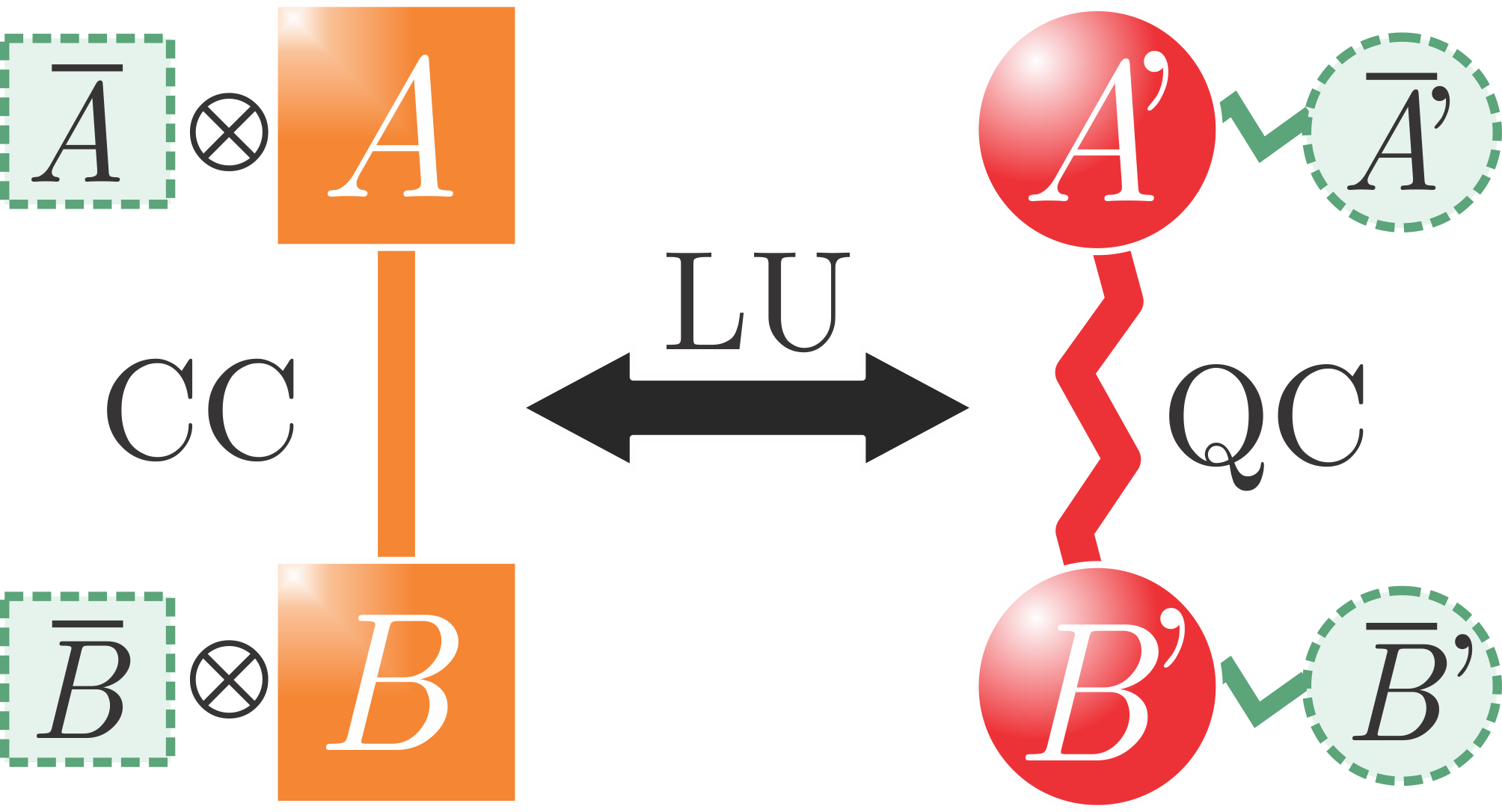}
    \caption{\label{fig:scheme_th1}(Color online) For a composite system $A+B$, CC with respect to that bipartition, coupling auxiliary uncorrelated local systems ($\bar{A}$ and $\bar{B}$) and performing LU operations generally produces a QC state (see Theorem~\ref{theo:qc1} for details). Here, the local unitaries have the form $U^{\bar{A}A}\otimes U^{B\bar{B}}$ and act by rearranging the local degrees of freedom, thus yielding `new' subsystems $A',\bar{A}',B',\bar{B}'$. In other words, the local unitaries induce new decompositions of $\bar{A}A$ and $B\bar{B}$ into different `primed' subsystems.}
    \end{figure}
    
Another way of assessing this important difference (between CC and non-correlated) with regards to the quantum correlations arising from the possible reduction of CC states exists: when the $A$ and $B$ subsystems have non-prime dimensions, it is possible to find reductions of $\rho^{AB}$ that respect the local bipartition $A|B$ and yet possess non-classical correlations. For example, the two-qubits' QC state $\rho^{A_1B_1} = p\ket{0^{A_1}}\bra{0^{A_1}}\otimes\ket{0^{B_1}}\bra{0^{B_1}}+(1-p)\ket{+^{A_1}}\bra{+^{A_1}}\otimes\ket{+^{B_1}}\bra{+^{B_1}}$ (the same as in the previous example) can be regarded as a reduction of the four-qubits' one $\rho^{AB} = p\ket{0^{A_2}}\bra{0^{A_2}}\otimes\ket{0^{A_1}}\bra{0^{A_1}}\otimes\ket{0^{B_1}}\bra{0^{B_1}}\otimes\ket{0^{B_2}}\bra{0^{B_2}} + (1-p)\ket{1^{A_2}}\bra{1^{A_2}}\otimes\ket{+^{A_1}}\bra{+^{A_1}}\otimes\ket{+^{B_1}}\bra{+^{B_1}}\otimes\ket{1^{B_2}}\bra{1^{B_2}}$. The latter is clearly CC with respect to the bipartition $A_1A_2|B_1B_2\cong A|B$. Thus, $\rho^{AB}$ is a CC state with QC reductions, always preserving the same prescription for the local degrees of freedom. Again, this result is very general: \textit{if $\rho^{AB}$ is a bipartite correlated state for which $A$ and $B$ are composites, then it is possible to find a reduction that possess quantum correlations.} We take this observation to be our second theorem.

    \begin{theorem} \label{theo:qc2}
Let $\rho^{AB}$ be a non-product density operator over $\hilb{A}\otimes\hilb{B}$, with $dim\hilb{A}$ and/or $dim\hilb{B}$ given by non-prime numbers. Then, it is possible to find a reduced state that preserves the local bipartition and possess quantum correlations.
    \end{theorem}
    
    \begin{proof}
The proof is rather trivial from our first Theorem. If $\rho^{AB}$ has the properties of the statement, we can regard $\rho^{AB}$ as the (already) extended state. Thus, applying local unitaries and tracing out some degrees of freedom yields the desired result. Being correlated is a necessary condition, since the reduction of any product state is trivially non-correlated with respect to a fixed bipartition.
    \end{proof}
    
As pointed out in the Introduction, the relative character of quantum correlations with respect to the chosen partition of a system into subsystems has been carefully studied in the case of pure states~\cite{Zana01,Zana04,Barn04,Hars11}. In our presentation, \textit{we focus on the case of mixed states under a locality restriction: we allow only \textit{local} unitary operations (over the enlarged Hilbert space) to explore the observables' subspaces of each local subsystem.}

Summing up, possession of CC states already implies some degree of quantumness in the correlations of both parts, in the following sense:
    \begin{itemize}
    \item CC states supplied with uncorrelated ancillas exhibit quantum correlations, in general,  when alternative local observables are specified and,
    \item reductions of CC states exhibit, in general, quantum correlations.
    \end{itemize}

    \begin{remark}
Local unitaries act over $\hilb{\bar{A}A}\otimes\hilb{B\bar{B}}$ by rearranging the local degrees of freedom to give an alternative decomposition into subsystems preserving the original local bipartition. For the reduced state $\rho^{AB}$, the transformation is equivalent to a local operation. The impact of local operations on quantum correlations has been seriously studied in the last years. Both properties of quantum correlations, stated in Theorems~\ref{theo:qc1} and~\ref{theo:qc2}, rely on the more general one that quantum correlations can be created by local noise (i.e., local quantum trace-preserving channels)~\cite{Stre11b,Hu11}. When one chooses this channel to be the trace operation, the above relations between classical and quantum correlations of composite systems arises.
    \end{remark}

We propose next a measure that attempts to quantify these facts from an information-theoretic point of view.

\section{Measuring the Potential Quantumness} \label{sec:PQ}
The two Theorems discussed so far refer to closely related facts that can be quantified by consideration of appropriate information-theoretic measures of quantum correlations. Given the previous analysis, we give a straightforward operational definition for our Potential Quantumness (PQ) measure.

    \begin{definition} \label{def:PQ}
Let $\rho^{AB}$ be a density operator over $\hilb{A}\otimes\hilb{B}$, and $\eta_0=\ket{0}\bra{0}$ the `ready' state over $\hilb{C}=\mathbb{C}^d$ of an auxiliary system. The PQ of $\rho^{AB}$, of rank $d$, with respect to the bipartition $A|B$ is
    \begin{equation}
    \PQ(\rho^{AB}) = \max_{U \in \text{LU}} Q^{A|B}(U\eta_0\otimes\rho^{AB}\otimes\eta_0\ctr{U}) \,,
    \end{equation}
where $Q^{A|B}(\rho):=Q(\text{Tr}_{\hilb{C}}\rho)$ implies tracing out the auxiliary systems, and $Q^{A|B}$ is any measure of bipartite quantum correlations between $A$ and $B$.
    \end{definition}
    
Usually, a measure $Q$ of quantum correlations is semi-definite positive, zero for CC states (in particular, for product states) and maximal for maximally entangled states. Those properties are fulfilled by every entanglement and discord-like measures. In those cases, the corresponding PQ-measure satisfies some basic properties that make it suitable as a measure of quantum correlations:
    \begin{itemize}
    \item (Positivity) $\PQ(\rho^{AB})\geq0$ for every state $\rho^{AB}$ and any dimension $d$ of the auxiliary parts.
    \item (Minimum) For any value of $d$, $\PQ(\rho^{AB})=0$ if and only if $\rho^{AB}=\rho^A\otimes\rho^B$.
    \item (Maximum) $\PQ$ is maximal for maximally entangled states.
    \end{itemize}
Positivity holds because $Q$ itself is semi-definite positive. Indeed, from Definition~\ref{def:PQ} follows the stronger relation $\PQ(\rho^{AB})\geq Q(\rho^{AB})$. Regarding the second property, $\PQ(\rho^A\otimes\rho^B)=0$ holds because the unitaries involved do not mix $A\bar{A}$ with $B\bar{B}$ degrees of freedom, and there is no LU that can correlate them, not even in a classical sense. On the other hand, if $\rho^{AB}$ is not a product state, then $\PQ(\rho^{AB})\neq0$%
\footnote{This can be proved by construction of a LU, $U$, such that $Q^{A|B}(U\rho^{ext}U^\dagger)\neq0$. Let us divide the cases of $\rho^{AB}$ being CC and QC. If it is QC, no extension is needed because $\rho^{AB}$ is already such that $Q(\rho^{AB})\neq0$. If, on the contrary, $\rho^{AB}$ is CC, we know that it has the form $\rho^{AB}=\sum_{ij}{p_{ij}\ket{i^A}\bra{i^A}\otimes\ket{j^B}\bra{j^B}}$, for $\{\ket{i^A}\}$ ($\{\ket{j^B}\}$) orthogonal states on $\hilb{A}$ ($\hilb{B}$). Consider the cases when only two of the probabilities are different from zero and, moreover, they involve different local projectors, i.e. $\rho^{AB}=p\ket{0^A}\bra{0^A}\otimes\ket{0^B}\bra{0^B}+(1-p)\ket{1^A}\bra{1^A}\otimes\ket{1^B}\bra{1^B}$. In this case, we can extend the state with ancillas of dimension $2$ on both sides: $\rho^{ext}=\ket{0^{\bar{A}}}\bra{0^{\bar{A}}}\otimes\rho^{AB}\otimes\ket{0^{\bar{B}}}\bra{0^{\bar{B}}}$. Now, $\ket{0^{\bar{A}}}\otimes\ket{0^A}$ and $\ket{0^{\bar{A}}}\otimes\ket{1^A}$ are orthogonal states on $\hilb{\bar{A}}\otimes\hilb{A}$, and can be mapped via a unitary transformation to any other pair $\{\ket{\alpha},\ket{\alpha_\perp}\}$ of orthogonal states. The same happens on $\hilb{B}\otimes\hilb{\bar{B}}$, where a local unitary transformation takes $\ket{0^B}\otimes\ket{0^{\bar{B}}}$ and $\ket{1^B}\otimes\ket{0^{\bar{B}}}$ into $\ket{\beta}$ and $\ket{\beta_\perp}$, respectively. Tracing out the ancillas yields $\rho^{A'B'}=p\rho^{A'}_{\alpha}\otimes\rho^{B'}_{\beta}+(1-p)\rho^{A'}_{\alpha_\perp}\otimes\rho^{B'}_{\beta_\perp}$, which is in general a QC state, unless $\rho^{A'}_{\alpha}$ and $\rho^{A'}_{\alpha_\perp}$ have common eigenbasis (and the same for $B'$).
If it is the case that both probabilities involve a common projector on one of the parties, the same procedure creates a QC state with respect to the other part. Generalization to higher rank states is straightforward.}.
The third property is fulfilled because $Q$ itself saturates for maximally entangled states, even when no extension or local operation is performed.

The defined measure exhibits many other interesting properties but, before presenting them, we prove the following theorem that provides an equivalent definition for $\PQ$, one that does not require any reference to auxiliary systems.

    \begin{theorem} \label{theo:PQ2}
For every state $\rho^{AB}$ and any dimension $d$ as in Definition~\ref{def:PQ}, the PQ of $\rho^{AB}$, of rank $d$, with respect to the bipartition $A|B$ is
    \begin{equation}
    \PQ(\rho^{AB}) = \max_{E\in \text{LO}(d)} Q(E[\rho^{AB}]) \,,
    \end{equation}
where $E\in \text{LO}(d)$ is any local operation of rank at most $d$, and $Q^{A|B}$ is any measure of bipartite quantum correlations between $A$ and $B$.
    \end{theorem}

    \begin{proof}
The equivalency is straightforwardly proven remembering that, by Stinespring's dilation theorem~\cite{Stin55}, any quantum operation can be reproduced by adding an ancilla, performing a unitary operation over the enlarged Hilbert space and finally tracing out the ancilla. In our case, the restriction on local unitaries impose the corresponding locality condition over the quantum operations of Theorem~\ref{theo:PQ2}.
    \end{proof}
    
Theorem~\ref{theo:PQ2} offers a new concise interpretation for our measure: $\PQ(\rho)$ quantifies the quantumness of the correlations of $\rho$ attainable by local operations. Among the vast family of local operations, we can identify, for example, the local unitaries, the local unitals and the local classical-quantum channels. As we show below, while the local unitaries do not change the value of $\PQ$, our measure is non-increasing under arbitrary local operations.

    \begin{remark}
The measure $\PQ$ depends on the value of $d$, which can be regarded as a restriction on the class of accessible local operations (LO) or a restriction on the dimension of the accessible auxiliary systems (see Definition~\ref{def:PQ}). As $\text{LO}(d'<d)\subset\text{LO}(d)$, it is straightforward to show that $\PQd{d'<d}(\rho^{AB})\leq\PQ(\rho^{AB})$ for any $\rho^{AB}$. However, it is interesting to consider the following particular $d$-independent scenario. If $\rho^{AB}$ is such that $\max\{dim(\hilb{A}),dim(\hilb{B}))\}=d_{\max}$, then any LO over that state can be implemented with auxiliary extensions of dimensions at most $d_{\max}^2$. Thus, $\PQ(\rho^{AB})\rightarrow\PQd{d_{\max}^2}(\rho^{AB})$ when $d\rightarrow\infty$. We shall refer to the later as the \textit{maximum-PQ} or $\mPQ$. It is clear that it accounts for the no-restriction-over-LO case. Finally, the case $d=0$ trivially matches the corresponding $Q$-measure, i.e. $\PQd{d=0}(\rho)=Q(\rho)$.
    \end{remark}
    
It is noteworthy that this measure, $\PQ$, does not (necessarily) involve a dynamical interpretation. Instead, the unitaries appearing in Definition~\ref{def:PQ} attempt to capture the relative character of the correlations with respect to the partition of a given system into subsystems. That is, if $\rho^{AB}$ is the state of a system (bi)partitioned according to $A|B$, and if there are auxiliary systems such that $\eta_0\otimes\rho^{AB}\otimes\eta_0$ is a possible joint state (as in Definition~\ref{def:PQ}), then the operations $U\in\text{LU}$ can be thought of as a resetting of the local subsystems (see Appendix~\ref{sec:PQobservables} for a more detailed discussion).

\subsection{Properties of the PQ-measure}
As expected, $\PQ$ inherits some particular properties of the chosen $Q$ measure. For example, if $Q$ is the usual QD then $\PQ$ is an asymmetric measure relying on one-partite measurements, while becoming symmetric if $Q$ is the symmetric QD. But, even before specializing things for a certain $Q$, we can prove further properties of the PQ-measure.

    \begin{theorem} \label{theo:pures}
For pure states $\rho^{AB}=\ket{\psi^{AB}}\bra{\psi^{AB}}$ and for any measure $Q$ that matches an entanglement monotone for pure states, the PQ-measure coincides with the corresponding entanglement monotone.
    \end{theorem}

    \begin{proof}
The proof is very simple. By Theorem~\ref{theo:PQ2}, $\PQ(\ket{\psi^{AB}})$ is the maximum value of $Q(E[\ket{\psi^{AB}}])$ over $E\in\text{LO}(d)$. If, i) for pure states,
$Q$ decreases monotonically under local operations and classical communication (LOCC), and ii) remembering that $\text{LO}\subset\text{LOCC}$, we have
    \begin{align*}
    \PQ(\ket{\psi^{AB}}) &= \max_{E\in\text{LO}(d)} Q(E[\ket{\psi^{AB}}]) \\
        &\leq \max_{E'\in\text{LOCC}} Q(E'[\ket{\psi^{AB}}]) \\
        &\leq Q(\ket{\psi^{AB}}) \,,
    \end{align*}
On the other hand, taking $E$ to be the identity map, we have that $Q(E[\ket{\psi^{AB}}])=Q(\ket{\psi^{AB}})$, saturating the above inequality. Accordingly, $\PQ(\ket{\psi^{AB}})=Q(\ket{\psi^{AB}})$ for every pure state $\ket{\psi^{AB}}$ and any value of $d$.
    \end{proof}
    
Next, we show that performing a LO over a certain state cannot increase the PQ unless the rank of the LO is greater than the one corresponding to the PQ-measure. The result follows from the optimization over LO in the definition of the PQ-measure.

    \begin{theorem} \label{theo:local_ops}
The PQ-measure of rank $d$ is non-increasing under a LO of rank equal or lower than $d$.
    \end{theorem}
    
    \begin{proof}
Let $E'\in\text{LO}(d')$ and $E'[\rho]=\rho'$ be the corresponding transformations of $\rho$, with $d'\leq d$. Then, for $\rho'$ it holds that     \begin{align*}
    \PQ(\rho') &= \max_{E\in\text{LO}(d)} Q(E[\rho']) \\
        &= \max_{E\in\text{LO}(d)} Q(E\circ E'[\rho]) \\
        &\leq \max_{E\in\text{LO}(d)} Q(E[\rho]) \\
        &= \PQ(\rho) \,,
    \end{align*}
where `$\circ$' indicates composition of operations. We used, in the third line, the fact that operations of the form $E\circ E'$ span a subset of $\text{LO}(d)$.
    \end{proof}
    
Of particular interest is the behaviour of any measure under LU operations. In this case, it is easy to show that the PQ-measure remains invariant.
    
    \begin{theorem} \label{theo:unitaries}
For any measure $Q$ of bipartite quantum correlations that is invariant under LU operations, the associated PQ-measure is also invariant under local unitaries.
    \end{theorem}
    
    \begin{proof}
Let $U^{A|B}$ be a unitary operation acting locally over $\hilb{A}\otimes\hilb{B}$. For any state $\rho$, the corresponding transformation is $\rho\mapsto\rho_1=U^{A|B}\rho U^{A|B}$. We want to compare the PQ of $\rho$ with that of $\rho_1$. From Definition~\ref{def:PQ},
    \begin{align*}
    \PQ(\rho_1) &= \max_{V\in LU} Q^{A|B}(V\eta_0\otimes\rho_1\otimes\eta_0\ctr{V}) \\
     &= \max_{V\in LU} Q^{A|B}(V U^{A|B}\eta_0\otimes\rho\otimes\eta_0 U^{A|B}\ctr{V}) \\
     &= \max_{V'\in LU} Q^{A|B}(V'\eta_0\otimes\rho_1\otimes\eta_0\ctr{V'}) \,,
    \end{align*}
which is equal to $\PQ(\rho)$. In the third line, we use that the composition of a unitary operation $V$ that is local over $\bar{A}A|B\bar{B}$ and another unitary operation $U^{A|B}$ that is local on $A|B$, yields a unitary $V'$ that is local over $\bar{A}A|B\bar{B}$.
    \end{proof}

As stated in Theorem~\ref{theo:qc2}, one would expect the PQ-measure to capture the quantum correlations from the possible reductions of a certain state. The following result shows that, indeed, such is the case.

    \begin{theorem} \label{theo:reductions}
Given a bipartite state $\rho^{AB}$ where $A$ and $B$ are also composites, $A=\{A_i\}$ and $B=\{B_j\}$, the PQ-measure of rank $d$ is lower bounded by the Q-measure over all the possible (fine-grained) reductions $\rho^{A_iB_j}=\text{Tr}_{\hilb{comp}}\rho^{AB}$, with $\hilb{comp}=\bigotimes_{m\neq i,n\neq j}\hilb{A_m}\otimes\hilb{B_n}$, such that $dim(\hilb{A_i})\times dim(\hilb{B_j}) \leq d$.
    \end{theorem}
    
    \begin{proof}
Any reduction $\rho^{A_iB_j}$ is the result of a LO of the form $\text{Tr}_{\hilb{A_i}}\circ\text{Tr}_{\hilb{B_j}}\circ\mathds{1}_{\hilb{comp}}$ over the state of the full system. Such an LO is thus included amongst those considered in the maximization-procedure involved in the definition of $\PQ$, if $dim(\hilb{A_i})\times dim(\hilb{B_j}) \leq d$.
    \end{proof}
    
As stated before, Theorem~\ref{theo:reductions} establishes that the PQ-measure takes into account the second fact mentioned in Section~\ref{sec:PQ}, namely that even CC states can have QC reductions, a phenomenon that is NOT captured by discord-like measures. For a given CC state, however, any reduction is a separable state for the same bipartition~\cite{Li08,Bell14,Bell15,Plas15}. As a consequence, the PQ of a given CC state is upper bounded by that of the separable states. From now on, we are going to concentrate efforts on the case of non-restricted local capabilities, for which the situation is well described by $\mPQ$. In that case, the theorems above asserts that:
    \begin{enumerate}
    \item[a)] $\mPQ$ matches an entanglement monotone for pure states and usual $Q$ measures,
    \item[b)] $\mPQ$ is non-increasing under LO,
    \item[c)] $\mPQ$ is invariant under LU,
    \item[d)] $\mPQ$ takes into account the quantum correlations of every possible reduction of the state.
    \end{enumerate}

\subsection{Relation to other measures of quantumness}
Some quantum correlations' measures involving local unitary operations have recently appeared in the Literature. They can be related, to some extent, to our above proposal.

Let us start by considering the interesting measure of quantumness advanced by Devi and Rajagopal (DR)~\cite{Devi08}. Given a bipartite state $\rho^{AB}$, they consider i) all the possible extensions $\rho^{\bar{A}AB}$ to a larger Hilbert space, such that $\text{Tr}_{\bar{A}}\rho^{\bar{A}AB}=\rho^{AB}$, and ii) the set of projective measurements over $\bar{A}A$. Hence, quantumness is defined as the minimum Kullback-Liebler relative entropy between the original state and the post-measurement state. As in the Potential discord (PD)-case, this measure involves an enlarged Hilbert space. However, since DR's measure is computed via a minimization, this quantity is expected to be lower than, for example, QD. Indeed, the authors have shown that their measure is an upper bound to the relative entropy of entanglement.

As a second example, we refer again to the work of Daki\'c \textit{et al.}~\cite{Daki10}, where the authors show that the rank of the correlation matrix of a bipartite state serves as a witness of quantum correlations. Any bipartite state can be written in terms of arbitrary bases $\{A_i\}$ and $\{B_j\}$ of Hermitian operators of the local Hilbert spaces, $\hilb{A}$ and $\hilb{B}$, as $\rho^{AB}=\sum_{ij}{r_{ij}A_i\otimes B_j}$. The number of non-zero singular values of the matrix $(r_{ij})$ is $L\leq d_{\min}^2$, with $d_{\min}=\min\{dim(\hilb{A}),dim(\hilb{B})\}$. For CC states, $L\leq d_{\min}$. Hence, $L\geq d_{\min}$ implies quantum correlations. Nonetheless, there are states with non-zero QD and $L<d_{\min}$. In particular, any separable state that can be created by local operations on CC states has $L<d_{\min}$. In our treatment, this implies that any CC state displays the same degree of quantumness that the most QC separable state that can be locally created from it (the latter has the same $L$ that the original CC state).

Finally, we look at another related work, due to Gharibian~\cite{Ghar12}, who defines a measure of nonclassicality as the minimal distance between the state and all its possible local unitary transformations. From our perspective, the LU operations accounts for a switch in local observables. Thus, Gharibian's measure captures the minimal disturbance suffered by a given state when changing the local observables. It turns out that Gharibian's measure is also a discord-like one, that is non-zero if and only if the state is not a CC state.
\vskip 4mm
\noindent None of the above measures captures what our measure of Potential Quantumness does, namely the non-classical correlations present in CC states. Next, so as to be able present some numerical results and obtain deeper insight into these matters, we specialize things by regarding QD as our quantum correlations measure.

\section{Quantum Potential Discord (PD)} \label{sec:PD}
Definition~\eqref{def:PQ} determines a family of correlation measures that depends on the particular functional $Q:\mathcal{L(H)}\rightarrow\mathbb{R}$ that one chooses to quantify the quantum correlations. For example, we can take $Q$ equal to the usual QD~\cite{Olli01,*Hend01}
    \begin{equation} \label{eq:discord}
    \delta(\rho) := \mathcal{I}(\rho)-\max_{\Pi}\mathcal{I}(\Pi[\rho]) \,,
    \end{equation}
with $\mathcal{I}$ the quantum mutual information and $\Pi[\rho]$ the post-local-measurement state. The measure $\delta$ attempts to capture the minimal disturbance suffered by the state under a local non-selective measurement. Also, QD is an essential resource in the performance of many quantum protocols~\cite{Cava11,Madh11,Fanc11,Stre11a,Daki12}.

Thus, Potential Discord (PD) should be defined as
    \begin{equation}
    \PD(\rho) := \max_{E\in\text{LO}} \delta(E[\rho]) \,,
    \end{equation}
where $\delta(\rho)$ is the usual QD given by Eq.~\eqref{eq:discord}. We are going to consider local measurements over $A$, i.e., $\Pi[\rho]=\sum_i{(\Pi^A_i\otimes\mathds{1}^B)\rho(\Pi^A_i\otimes\mathds{1}^B)}$ for some local projective measurement $\{\Pi^A_i\}$. Analogous results can be found using bi-local measurements, or considering generalized (instead of projective) measurements. From its definition, it is straightforward to observe that PD is an intermediate measure between QD and mutual information: $\delta\leq\PD\leq\mathcal{I}$. In particular, for pure states, both QD and PD collapse to the Entropy of Entanglement and one has $S(\rho^A)=\delta=\PD\leq\mathcal{I}=2S(\rho^A)$, with $S(\rho^A)=S(\rho^B)$ the von Neumann entropy of the reduced state. On the other hand, for CC states, QD vanishes and one has $0\leq\PD\leq\mathcal{I}$. Finally, all three measures are zero for product states.

In order to understand the role played by the maximization procedure involved in a PD-computation, consider the one-parameter family of fully CC states $\rho^C_{\eta}=(1-\eta)\ket{00}\bra{00}+\eta\ket{11}\bra{11}$, with $0\leq\eta\leq1$. (Hereon, in order to simplify notation, we use $\ket{ij}$ instead of $\ket{i^A}\otimes\ket{j^B}$.) For any value of $\eta$, the state is a mixture of product and mutually orthogonal (i.e. fully distinguishable) states. Thus, $Q(\rho^C_{\eta})=0$ for any known measure of quantum correlations $Q$. However, for any $\eta\not\in\{0,1\}$, a local operation will create quantum correlations with respect to the bipartition. As seen in Fig.~\ref{fig:potdis_class}, PD captures this idea, distinguishing the quantumness of the different members of the $\rho^C_{\eta}$-family depending on the amount of quantum correlations that can be created under a LO. Highly symmetric families given by isotropic and Werner states has the same amount of PD than QD. We parametrize isotropic states as $\rho^I_{\eta}=(1-\eta)(\mathds{1}/3)+\eta\ket{\beta}\bra{\beta}$, with $\ket{\beta}$ a Bell-type state, and Werner states are given by $\rho^W_{\eta}=(\eta/3)P_++(1-\eta)P_-$, with $P_\pm=(\mathds{1}\pm\mathbb{P})/2$ and $\mathbb{P}=\sum_{ij}{\ket{ij}\bra{ji}}$.

    \begin{figure}
    \centering
    \includegraphics[width=.99\columnwidth]{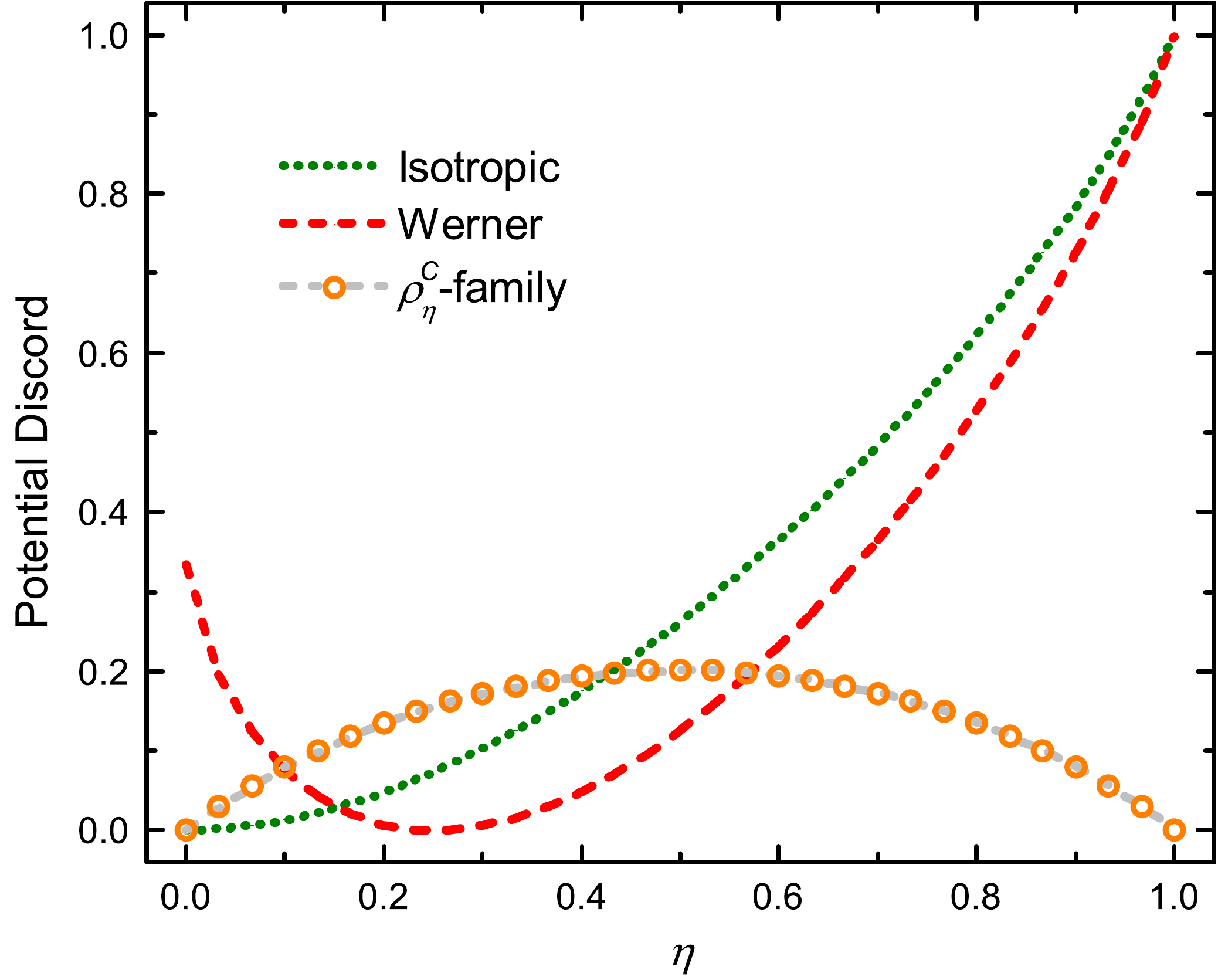}
    \caption{\label{fig:potdis_class}(Color online) Classically-correlated states can exhibit some quantumness revealed by PD. That is the case for the $\rho^C_\eta$-family of fully CC states (orange circles), which attains maximal PD when $\eta=1/2$. Isotropic (green dotted line) and Werner (red dashed line) families have a PD that equals their QD, as no LO can increase quantum correlations for them.}
    \end{figure}

    \begin{figure}
    \centering
    \includegraphics[width=.99\columnwidth]{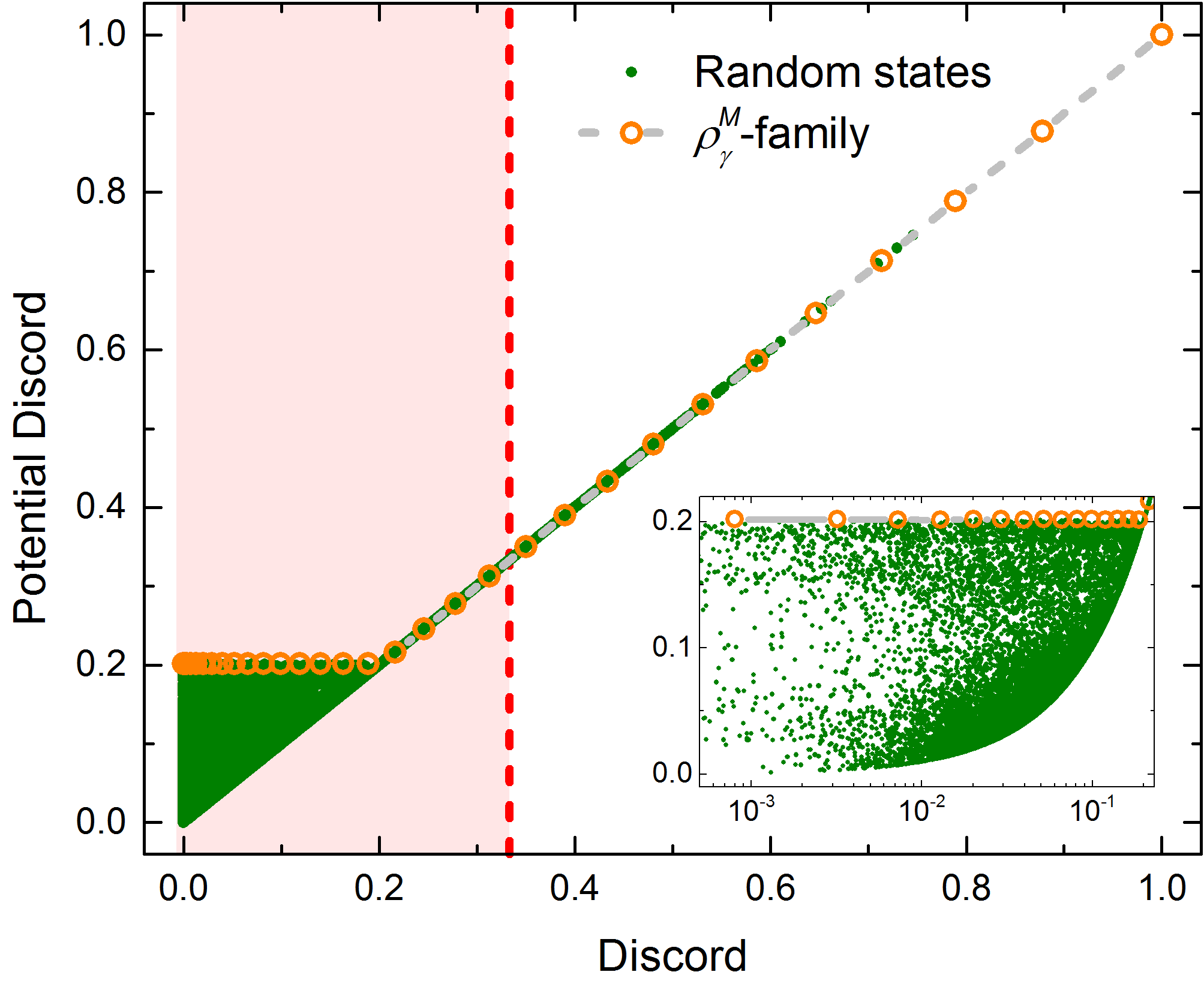}
    \caption{\label{fig:potdis_rndm}(Color online) Any separable state with QD $\leq0.2018$ can be reached from a CC state of two qubits, performing local operations. Nonetheless, there are separable states with QD above $0.2018$. The comparison between QD and PD displays this behaviour. The upper bound is given by the $\rho^M_\gamma$-family (orange circles; see text for details). Green data points correspond to $\sim10^5$ randomly generated states. The (red-shaded) region at the left of the vertical dashed line corresponds to the values of QD achievable by separable states. Inset: detail of the bottom leftmost region.}
    \end{figure}
    
    \begin{figure}
    \centering
    \includegraphics[width=.99\columnwidth]{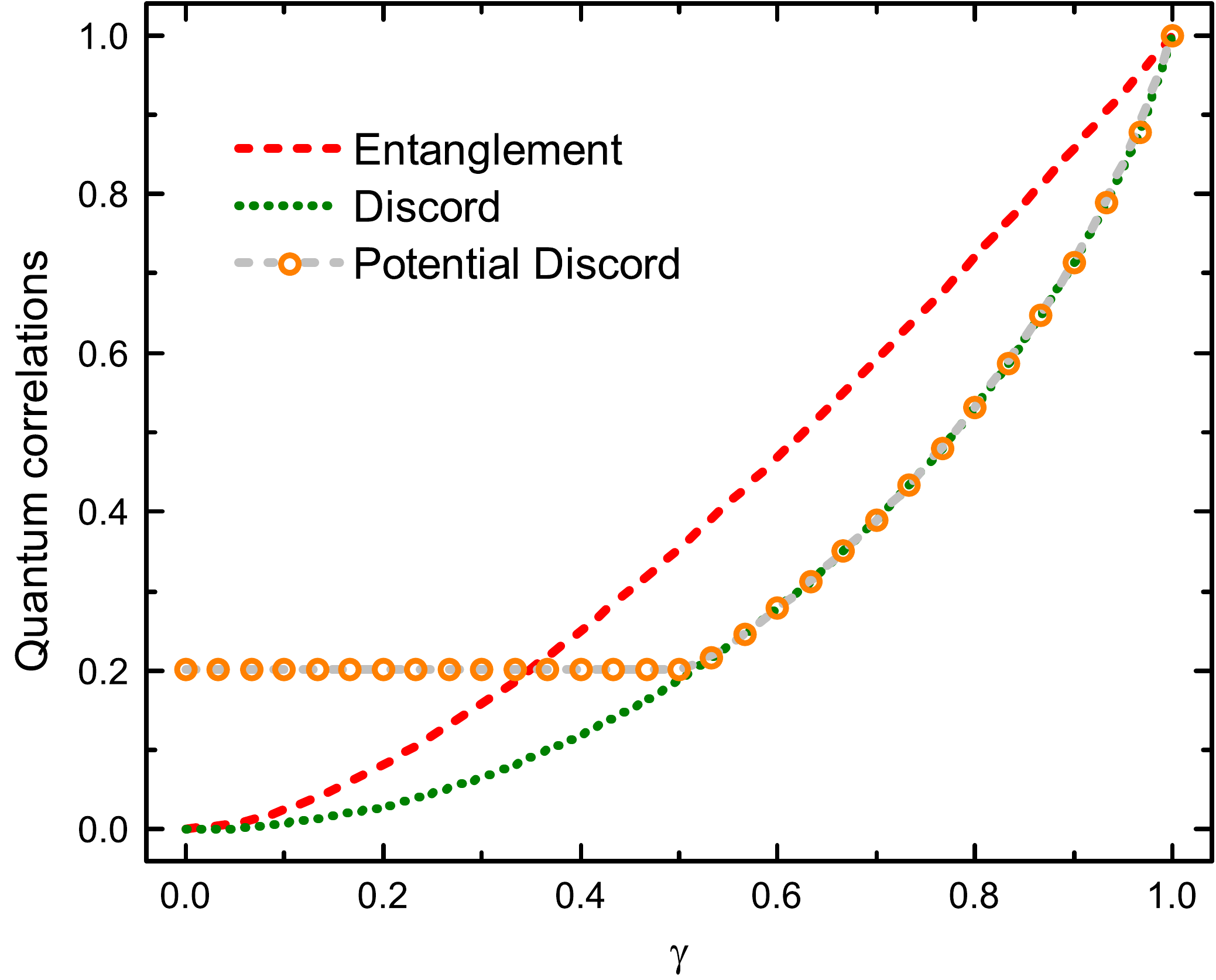}
    \caption{\label{fig:rho_gamma}(Color online) Different measures of quantum correlations for the $\rho^M_\gamma$-family. Entanglement of formation (red dashed line) is finite for any $\gamma>0$. QD (green dotted line) behaves in the same qualitative way as entanglement, while PD (orange circles) exhibits a transition: below $\gamma=1/2$ PD takes a constant minimum value and grows monotonically with $\gamma$ for $\gamma\geq1/2$.}
    \end{figure}

In Fig.~\ref{fig:potdis_rndm}, the upper bound is given by mixtures of a rank 2 CC state, $\rho_{class}=(\ket{00}\bra{00}+\ket{11}\bra{11})/2$, with the maximally-entangled state $\ket{\beta}$, i.e., by the family $\rho^M_\gamma=(1-\gamma)\rho_{class}+\gamma\ket{\beta}\bra{\beta}$, with $0\leq\gamma\leq1$. Separable states with QD above $0.2018$ cannot be reached from CC states by LO, as a direct consequence of i) the nature of the 
maximal QD of separable states with fixed rank~\cite{Bell15} and ii) the local creation of quantum correlations from CC states~\cite{Gess12}. Fig.~\ref{fig:rho_gamma} shows how entanglement, QD, and PD behave differently for the states $\rho^M_\gamma$. The left border of Fig.~\ref{fig:potdis_rndm} is reproduced by mixtures of rank 2 CC states with the maximally mixed state. Finally, the lower bound is given by isotropic states.

It is important to note that, as seen in Fig.~\ref{fig:potdis_rndm}, \emph{for those states of two-qubits with QD above $0.2018$, arbitrary local operations cannot increase their quantumness.} An open question is whether for arbitrary quantum states there is a lower bound of quantumness above which there is no local operation that increases their quantumness.

\subsection{Potential Discord under local noise}
In Ref.~\cite{Cicc12}, the authors show how local noise ---in particular, a local Markovian amplitude-damping channel (AD)--- can enhance quantum correlations for a two-qubit system initially in the fully CC state, $\rho_0=(\ket{+0}\bra{+0}+\ket{-1}\bra{-1})/2$. The authors find that, when the state is transformed under the local AD channel, QD reaches a maximal value of $0.07$. On the other hand, the state $\rho_0$ has maximal PD amongst all CC states, $\PD(\rho_0)=0.2018$, as it can be determined by a local unitary transformation of $\rho^C_{\eta}$ with $\eta=1/2$ (see Fig.~\ref{fig:potdis_class}). An immediate conclusion follows: AD is not optimal among local operations that create quantumness on $\rho_0$. Another way to understand the situation is that the authors are computing the PD restricted to a subset of LO, namely, the local AD operations. Imposing such restriction leads to a maximal PD of $0.07$. Taking $E_0=\ket{0}\bra{0}+\sqrt{1-p}\ket{1}\bra{1}$ and $E_1=\sqrt{1}\ket{0}\bra{1}$ as the corresponding Kraus operators for the AD channel, with $p=1-\text{e}^{-\Gamma t}$ and $\Gamma$ the relaxation rate, QD selects an intermediate value of $p$ as the one that maximizes quantumness, while PD sets $p=0$ (when no operation is performed at all) as the one of maximum quantumness (see Fig.~\ref{fig:amp_damp}).

    \begin{figure}
    \centering
    \includegraphics[width=.99\columnwidth]{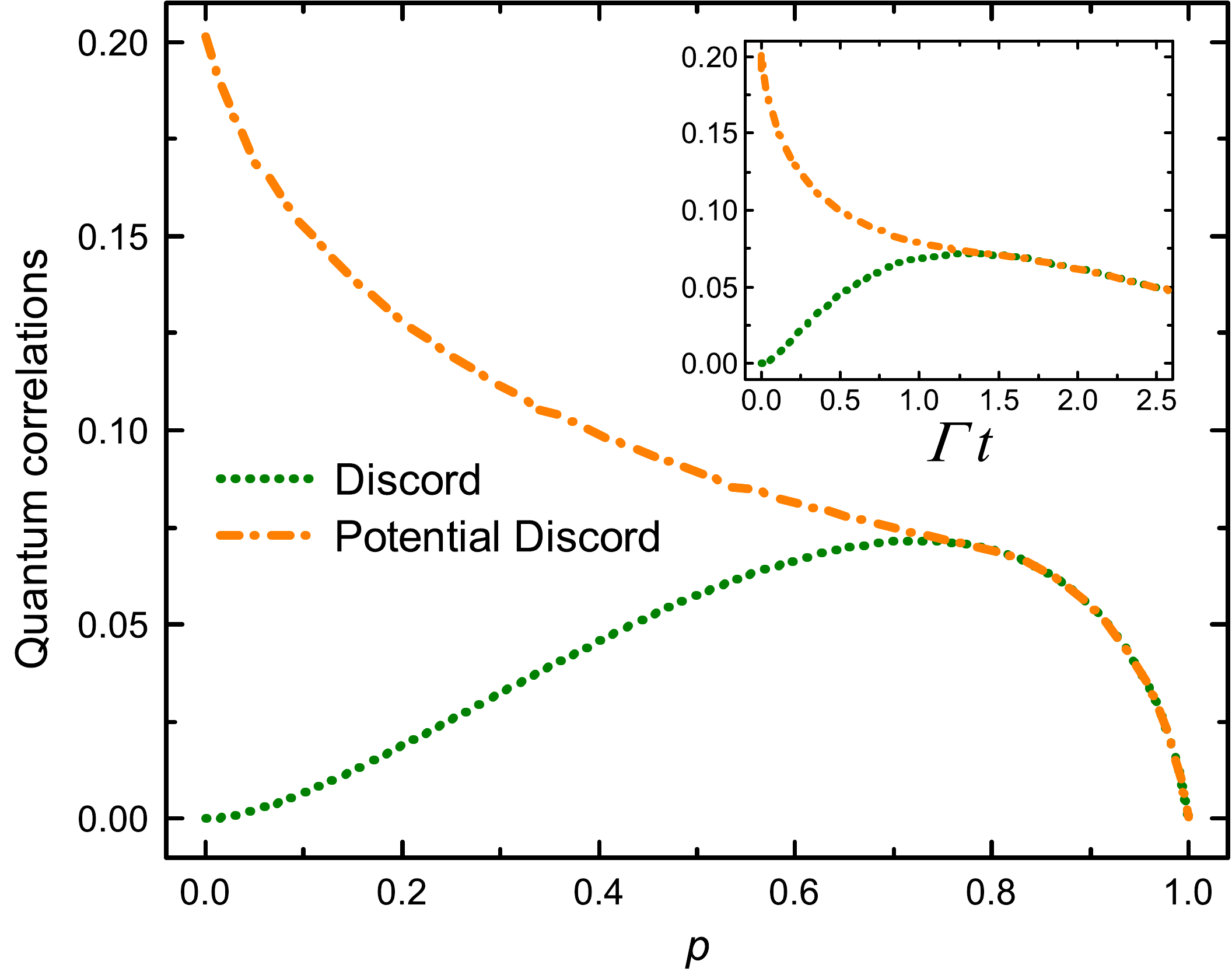}
    \caption{\label{fig:amp_damp}(Color online) Under an amplitude damping channel, QD (green dotted line) and PD (orange dashed-dotted line) exhibit different behaviours, indicating the fact that both measures reflects different aspects of the quantumness of a given state. Inset: QD and PD as a function of time (in units of the inverse of the relaxation rate $\Gamma$).}
    \end{figure}

\section{Conclusions} \label{sec:Conc}
Summing up, we have investigated the relative character of quantum correlations in bipartite states with respect to the local observables of both parts. The question is closely related to that of the effect of local operations on quantum correlations. We have proposed a measure of quantumness that takes into account this relative character. This novel quantifier involves a maximization procedure over any measure of bipartite quantum correlations that, in principle, makes it hard to compute it---considering that the usual measures of correlations involve an optimization procedure too. However, in some low-dimensional or highly symmetrical situations, our quantumness measure can be computed by taking advantage of known results. In particular, we have applied the measure to special families of states of two qubits. Also, we have compared our quantumness measure' behaviour with that Quantum Discord for the case of two qubits evolving under the effect of a local amplitude-damping channel.

We stress the fact that our presentation is not in contradiction with the one of Dakic \textit{et al.}~\cite{Daki10} and Gessner \textit{et al.}~\cite{Gess12}, who state that quantum correlations of those separable states that can be produced from CC states by local operations are not genuinely quantum. Indeed, our results highlight this fact: we observe a collapse of the quantum correlations of those separable states to a constant value of what we call the Potential Quantumness. Moreover, we point out that this quantumness-degree is already present in those states, without (and before) considering any kind of operation, as can be confirmed by measuring the appropriate local observables.

Finally, we summarize the following points, whose connections are reflected by our discussion:
    \begin{itemize}
    \item the relativity of quantum correlations with respect to the way in which the system is decomposed into subsystems; in particular, when a locality condition is imposed, this relative character depends on the preferred observables of each part;
    \item the effect of local operations on quantum correlations;
    \item the criterion of the rank of the correlation matrix to detect quantumness proposed by Daki\'c, Vedral and Brukner~\cite{Daki10}.
    \item the relation between separable and classically-correlated states, pointed out by Li and Luo~\cite{Li08}.
    \end{itemize}

\begin{acknowledgments}
This work was supported by CONICET-Argentina.
\end{acknowledgments}

\appendix

\section{Quantum correlations under global unitary operations} \label{sec:QCunderU}
When we remove the locality restriction, we are allowed to explore the whole observables' space. This situation was previously studied by Zanardi~\cite{Zana01}, where he distinguishes between \textit{virtual} and real subsystems. Also, in Zanardi \textit{et al.}~\cite{Zana04}, the authors studied the role of the relevant observables in the tensor product structure of the Hilbert space. Harshmann and Ranade~\cite{Hars11} gave a formal proof of the fact that, in the pure states case, one can tailor the observables so as to go from a situation with no entanglement to a maximal entanglement one, for any fixed state. However, for mixed states the situation is radically different ---as it can be seen, for example, taking the maximally mixed state, which is the same for any observables we choose--- and the tailoring of observables cannot place all the states on an equal footing.

Here, an important role is played by pseudo-pure states (weighted mixtures of a maximally mixed state with a pure state), namely $\rho_{a,\psi}=(1-a)(\mathds{1}/d)+a\ket{\psi}\bra{\psi}$, with $0\leq a\leq1$. For a fixed value of $a$, one can maximize the QD by taking $\ket{\psi}=\ket{\beta}$ as the maximally entangled state, so as to convert $\rho_{a,\psi}$ into an isotropic state. This can be accomplished applying a (global) unitary operation over $\ket{\psi}$, i.e.,  there is always a unitary $U_{\psi}$ such that $U_{\psi}\ket{\psi}=\ket{\beta}$. Hence, $U_{\psi}\rho_{a,\psi}U_{\psi}^{\dagger}=(1-a)(\mathds{1}/d)+a\ket{\beta}\bra{\beta}=\rho_{a,\beta}$.

    \begin{figure}
    \centering
    \includegraphics[width=.99\columnwidth]{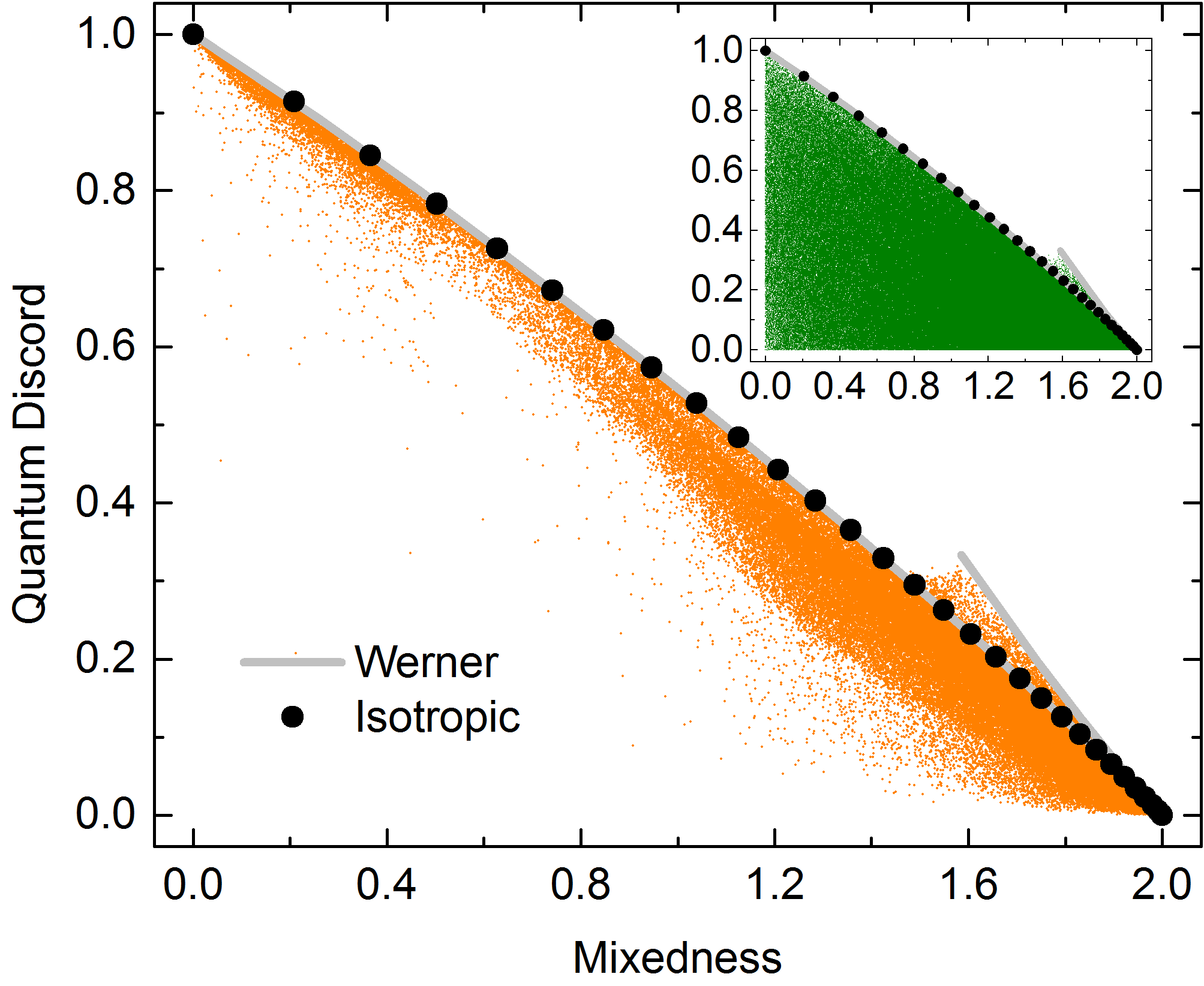}
    \caption{\label{fig:dis_underU}(Color online) Under global unitary operations, QD of a fixed state can be increased by an amount depending on the mixedness of the state. In particular, pseudo-pure states can be transformed into isotropic states (black dots), maximizing the QD for a certain value of entropy. Werner states (grey line) provide the upper bound for almost every value of entropy. Orange data points correspond to $\sim10^5$ random states. Inset: original QD for the same states. The region between isotropic and Werner states is upper bounded by families of two parameters (see Ref.~\cite{Giro11} for details).}
    \end{figure}

\section{Observables and the structure of the Hilbert space} \label{sec:PQobservables}
The simplest case may be a four qubits system, where $A$ and $B$, with $\hilb{A}\cong\mathbb{C}^4\cong\hilb{B}$, are subsystems of two qubits. The state-space is determined by density operators over $\hilb{AB}\cong(\mathbb{C}^{2})^{\otimes{4}}\cong(\mathbb{C}^{4})^{\otimes{2}}\cong\mathbb{C}^{16}$. If we impose a \textit{locality} restriction for the ${A|B}$ partition, then any unitary operation ${U^{A|B}=U^A{\otimes}U^B}$ over ${\mathbb{C}^{4}\otimes\mathbb{C}^{4}}$ acts locally over the $A$ and $B$ degrees of freedom. Thus, the action of ${U^A{\otimes}U^B}$ can be interpreted in two alternative ways:
    \begin{itemize}
    \item[a)] as a bi-local unitary transformation over the space of states;
    \item[b)] as a transformation over the spaces of local observable operators, that is a reconfiguration of the local degrees of freedom.
    \end{itemize}
Indeed, for any state $\rho^{AB}$ and any observable $O$, its expectation value is
    \begin{align*}
    \braket{O}_{U^{A|B}\rho^{AB}\ctr{U^{A|B}}} &=
    \text{Tr}[{(U^{A|B}\rho^{AB}\ctr{U^{A|B}})O}] \\
    &= \text{Tr}[{\rho^{AB}(\ctr{U^{A|B}}O U^{A|B})}] \\
    &= \braket{\ctr{U^{A|B}}O U^{A|B}}_{\rho^{AB}} \,.
    \end{align*}

What determines that ${\rho^{AB}}$ has subsystems $A$ and $B$? The determination of these subsystems is certainly not unique and usually relies on the accessible degrees of freedom of our joint system. A \textit{qubit} is an abstract entity, well-suited for the description of quantum bi-stable systems, as e.g. a spin one-half particle. For two independent particles for which the only relevant degrees of freedom are thier one-half spins, the `natural' description is given by a density operator over ${\mathbb{C}^2\otimes\mathbb{C}^2}$. The situation is better understood from the \textit{observables perspective}. The natural observables are the spin operators in $A$ and $B$, represented by the corresponding Pauli matrices ${\sigma^A_i}$ and ${\sigma^B_i}$, with $i=x,y,z$. If they represent the relevant degrees of freedom, the description of our system can be given in terms of the algebra spanned by them, ${\mathcal{O}=\text{span}\{\mathds{1}\otimes\mathds{1},\sigma_i\otimes\mathds{1},\mathds{1}\otimes\sigma_j\}}$, where $i,j=x,y,z$ and ${\mathds{1}}$ is the identity operator for $\mathbb{C}^2$. We removed the $A$ and $B$ superscripts so as to simplify the notation. The algebra $\mathcal{O}$ induces a tensor product structure over the Hilbert space of the joint system, i.e. ${\hilb{AB}=\hilb{A}\otimes\hilb{B}}$~\cite{Zana04}. But these `natural observables' are not necessarily natural at all: any unitary transformation ${\mathcal{O}\mapsto\mathcal{O}_U:=U\mathcal{O}\ctr{U}}$ could a priory be regarded,  without any further physical assumption, on an equal footing with the presumed natural spin-one.

For example, the pure non-correlated state $\ket{\psi}=\ket{00}$ are unitarily equivalent to the maximally entangled state $\ket{\beta_+}=(\ket{00}+\ket{11})/\sqrt{2}$ via a transformation $U_\psi$. Alternatively, one can assert that $\ket{\beta_+}$ is the representation of $\ket{\psi}$ in terms of the algebra $U_\psi\mathcal{O}\ctr{U_\psi}$ of observables. Harshman and Ranade gave in Ref.~\cite{Hars11} the formal proof that, for any pure state in $\mathbb{C}^N$, where $N\in\mathbb{N}$ is not a prime, the observables can be `tailored' to induce a subsystem decomposition for which the state has a desired level of entanglement, from a product state to a maximally entangled one. The situation is rather different for mixed states.

When dealing with mixed states, the unitaries do not allow one to surf the whole space of states. In particular, a unitary transformation can not change the eigenvalues of a given state, thus preserving its original entropy. Let us suppose that we start with the mixed separable ---indeed, CC--- state $\rho^{AB}_p=p\ket{00}\bra{00}+(1-p)\ket{11}\bra{11}$. A unitary transformation ${\rho^{AB}_p\mapsto U\rho^{AB}_p\ctr{U}}$ will preserve the purity and orthogonality of both components. For example, applying the ${U_\psi}$ defined above, we have $U_\psi\rho^{AB}_p\ctr{U_\psi} = p \ket{\beta_+}\bra{\beta_+} + (1-p) \ket{\beta_-}\bra{\beta_-}$, with ${\ket{\beta_\pm}:=(\ket{00}\pm\ket{11})/\sqrt{2}}$ two orthogonal maximally correlated states. For $p=0$ or $p=1$, $\rho^{AB}_p$ is pure and its transformed version becomes a Bell-type state. When, ${p\in(0,1)}$, $\rho^{AB}_p$ is never pure and its transformed version is not maximally entangled.



\bibliography{quares}

\end{document}